\documentclass[11pt]{article}

\usepackage[margin=1in]{geometry}
\usepackage{amsmath,amssymb,amsthm,mathtools}
\usepackage{enumitem}
\usepackage{microtype}
\usepackage{hyperref}
\hypersetup{colorlinks=true,linkcolor=black,citecolor=black,urlcolor=black}

\newtheorem{theorem}{Theorem}[section]
\newtheorem{lemma}[theorem]{Lemma}
\newtheorem{proposition}[theorem]{Proposition}
\newtheorem{corollary}[theorem]{Corollary}
\newtheorem{definition}[theorem]{Definition}
\newtheorem{assumption}[theorem]{Assumption}
\newtheorem{remark}[theorem]{Remark}

\newcommand{\A}{\mathcal A}
\newcommand{\B}{\mathcal B}
\newcommand{\D}{\mathcal D}

\newcommand{\F}{\mathcal F}
\newcommand{\Hh}{\mathcal H}
\newcommand{\I}{\mathfrak I}
\newcommand{\M}{\mathcal M}
\newcommand{\N}{\mathcal N}
\newcommand{\R}{\mathcal R}

\DeclareMathOperator{\CP}{CP}
\DeclareMathOperator{\Ev}{Ev}

\DeclareMathOperator{\MD}{MD}
\DeclareMathOperator{\Proj}{Proj}
\DeclareMathOperator{\Stab}{Stab}
\DeclareMathOperator{\Tr}{Tr}
\DeclareMathOperator{\vN}{vN}

\title{Emergence of Boolean Facts from Markovian Coarse-Graining in Relational Quantum Causal Processes}
\author{Yipeng Xu\thanks{Correspondence: \texttt{yipeng.xu@ieee.org}.}\\
\small University of Nottingham Ningbo China}
\date{}

\begin{document}
\maketitle

\begin{abstract}
We formulate an operator-algebraic mechanism by which exact Boolean records can arise from local completely positive quantum operations without being imposed as microscopic structure.  The kinematic input is an algebraic process functional assigning probabilities to local normal completely positive operations in a finite operational context.  From the predual response of a target algebra to source interventions, relative to a background strategy class, we define an influence algebra; exact events are then, by definition, the projections in its center.  The dynamical question is whether nontrivial centers can be generated by coarse-graining rather than inserted through split-record laboratories.

We address this question using state-preserving normal unital completely positive coarse-graining channels.  If the Cesaro means of such a channel converge to a Choi--Effros infrared range and the range is asymptotically abelian in the GNS seminorm, then the represented infrared algebra is a commutative von Neumann algebra.  Its projection lattice is therefore a complete Boolean algebra.  We also give a finite-sector block-primitive criterion, motivated by locality and scrambling, which implies this asymptotic abelianness with exponential suppression of off-sector coherences and intra-sector fluctuations.  The result is a conservative mathematical statement: classical facts are not identified with arbitrary projections of a Type-III local algebra, but with central projections selected by an asymptotically abelian completely positive infrared limit.
\end{abstract}



\section{Introduction}

A basic requirement on any background-independent quantum theory is that classical records should not have to be inserted as prior microscopic axioms.  The theory should explain why some infrared propositions are jointly readable and obey Boolean logic, while most microscopic quantum propositions do not.  This issue is especially sharp in algebraic quantum field theory (AQFT), where local observable algebras are typically Type-III von Neumann factors.  Such factors have no minimal projections and have trivial center, so exact macroscopic records cannot be obtained merely by declaring arbitrary local projections to be classical events.

The purpose of this paper is not to propose a new interpretation of measurement.  It is more limited, and correspondingly more mathematical: we isolate a completely positive coarse-graining mechanism that produces a commutative infrared von Neumann algebra from a noncommutative quantum process.  The physical motivation is the emergence of macroscopic facts in relational quantum causal processes (RQCP), but the main statements are about normal unital completely positive (UCP) maps, Choi--Effros ranges, stabilizer-defined influence algebras, and asymptotic abelianness in a GNS seminorm.

The starting point is operational.  A finite context is described by local von Neumann algebras and by a normal positive multilinear process functional on local completely positive operations.  Interventions on one subsystem induce normal predual response differences on another subsystem.  The unitaries that leave all such response differences invariant define a stabilizer, and the commutant of this stabilizer is the background-relative influence algebra.  We define exact events as projections in the center of this influence algebra.  This definition is deliberately restrictive: noncommuting effects may be operationally meaningful, but they are not exact Boolean facts.

The central dynamical question is then precise.  Given a noncommutative microscopic algebra $\A$ and a normal UCP coarse-graining channel $\Gamma:\A\to\A$, when does its infrared range define a commutative algebra of exact records?  Choi--Effros theory shows that the range of a completely positive projection carries a natural $C^*$-product.  However, a Choi--Effros range need not be commutative.  The missing condition is asymptotic abelianness: after coarse-graining, the surviving infrared observables must commute in the state-dependent representation relevant for macroscopic readout.

Our first main result gives the corresponding structural theorem.  Let $\omega$ be a faithful normal state preserved by $\Gamma$, and suppose the Cesaro averages
\[
E_N(A)=\frac1N\sum_{k=0}^{N-1}\Gamma^k(A)
\]
converge to a normal UCP projection $E_\infty$.  If the infrared range $E_\infty(\A)$ is asymptotically abelian in the GNS seminorm, then the von Neumann algebra $\pi_\omega(E_\infty(\A))''$ is commutative.  Consequently its projections form a complete Boolean algebra.  In the RQCP terminology, exact macroscopic facts are the center projections of this represented infrared influence algebra.

Our second main result gives a concrete sufficient mechanism.  Suppose the channel admits mutually orthogonal invariant macroscopic sector projections $P_a$, is primitive within each block, and suppresses off-block coherences exponentially.  Then
\[
\Gamma^n(X)=\sum_a\omega_a(P_aXP_a)P_a+R_n(X),
\qquad
\|R_n(X)\|_\omega\le C_Xe^{-\gamma n}.
\]
The limiting algebra is generated by the sector projections $P_a$ and is therefore abelian; the same estimate gives exponential asymptotic abelianness.  The assumptions abstract the role of locality, finite propagation, and scrambling/ETH-type mixing: local dynamics restricts the spread of information, while primitive mixing removes microscopic phase data inside each macroscopic sector.  Protected nonabelian noiseless subsystems are then explicit failure modes rather than hidden exceptions.

The paper is organized as follows.  Section~\ref{sec:events} defines algebraic process functionals, influence algebras, and exact events.  Section~\ref{sec:ce} proves the Choi--Effros/GNS theorem for asymptotically abelian infrared ranges.  Section~\ref{sec:block} gives the block-primitive sufficient condition.  Section~\ref{sec:split-rg} explains how split-record laboratories and multiplicative-domain RG maps appear as finite-stage approximations rather than fundamental assumptions.  Section~\ref{sec:limits} summarizes the scope and failure modes of the construction.

\section{Operational framework and algebraic events}
\label{sec:events}

This section fixes the kinematic language.  No causal order, spacetime metric, or pre-existing Boolean sample space is assumed.

\subsection{Algebraic process functionals}

\begin{definition}[Algebraic process functional]
Let $C=\{\N_i\}_{i=1}^n$ be a finite family of local von Neumann algebras.  An algebraic process functional is a map
\[
\Omega_C:\prod_{i=1}^n \CP_\sigma(\N_i)\longrightarrow [0,1]
\]
which is normal and multilinear in each argument, positive on normal completely positive operations, and normalized on deterministic normal UCP operations:
\[
\Omega_C(\Phi_1,\ldots,\Phi_n)=1
\]
whenever all $\Phi_i$ are deterministic.
\end{definition}

Here $\CP_\sigma(\N_i)$ denotes normal completely positive maps on $\N_i$.  In finite-dimensional process-matrix theory one has
\[
p(\vec a|\vec x)=\Tr\left[W_C\bigotimes_{i\in C}M_i^{a_i|x_i}\right],
\]
where $W_C$ is a positive process matrix normalized on deterministic operations \cite{OCB2012}.  The preceding definition is the corresponding von Neumann algebraic abstraction.

\subsection{Background-relative influence algebras}

Fix a target von Neumann algebra $\M_Y$ and a source system $X$ in a context $C$.  A background strategy class $\B$ specifies the allowed operations on the complement of $X\cup Y$.  For $\beta\in\B$ and a source operation $\Phi_X$, define the normal response functional on $\M_Y$ by
\[
r_Y^{C;\beta}(\Phi_X)(A)=\Omega_C(\Phi_X,A,\beta),
\qquad A\in\M_Y.
\]
Response differences are
\[
\delta_Y^{C;\beta}(\Phi_X,\Phi'_X)
=r_Y^{C;\beta}(\Phi_X)-r_Y^{C;\beta}(\Phi'_X).
\]
Let $\D_Y^{C;\B}$ be the family of all such normal linear functionals.

\begin{definition}[Influence algebra]
For a family $\D$ of normal linear functionals on a von Neumann algebra $M$, set
\[
\Stab(\D)=
\{U\in\mathcal U(M):\delta(U^*AU)=\delta(A),\ 
\forall A\in M,\ \forall\delta\in\D\}.
\]
The influence algebra generated by $\D$ is
\[
\I(\D)=\Stab(\D)'\cap M.
\]
For the process context above we write
\[
\I_Y^{C;\B}:=\I(\D_Y^{C;\B})\subseteq \M_Y .
\]
\end{definition}

Thus $\I_Y^{C;\B}$ is the algebra of target observables that can be changed, or distinguished, by the specified source interventions after quotienting by background symmetries of the response data.

\begin{proposition}[Finite-dimensional density form]
\label{prop:finite-density}
Let $M=M_d(\mathbb C)$ and suppose $\delta_k(A)=\Tr(\Delta_kA)$ for matrices $\Delta_k$.  Then
\[
\I(\D)=\vN\{\Delta_k,\Delta_k^*:k\}.
\]
\end{proposition}

\begin{proof}
The stabilizer condition is
\[
\Tr(\Delta_kU^*AU)=\Tr(\Delta_kA)
\qquad\forall A,k .
\]
By cyclicity and nondegeneracy of the trace pairing, this is equivalent to
$U\Delta_kU^*=\Delta_k$ for every $k$, and similarly for $\Delta_k^*$.
Hence $\Stab(\D)=\mathcal U(\vN\{\Delta_k,\Delta_k^*\}')$, and taking the commutant inside $M_d(\mathbb C)$ gives the claim.
\end{proof}

\subsection{Exact events as center projections}

\begin{definition}[Exact event]
An exact event is a tuple
\[
e=(C,Y,\B,P),
\qquad
P\in\Ev_0(Y|C;\B),
\]
where
\[
\Ev_0(Y|C;\B):=\Proj\bigl(Z(\I_Y^{C;\B})\bigr).
\]
\end{definition}

The point of using the center is that it separates exact facts from merely sharp but incompatible quantum propositions.

\begin{proposition}[Boolean facts]
\label{prop:boolean-center}
For every von Neumann algebra $N$, the projection lattice $\Proj(Z(N))$ is a complete Boolean algebra.  Its finite lattice operations are
\[
P\wedge Q=PQ,\qquad
P\vee Q=P+Q-PQ,\qquad
\neg P=1-P .
\]
\end{proposition}

\begin{proof}
$Z(N)$ is a commutative von Neumann algebra.  Projections in a commutative von Neumann algebra commute and hence satisfy the displayed lattice operations.  Completeness follows because arbitrary suprema of projections exist in the strong operator topology and the center is strongly closed.
\end{proof}

\begin{remark}[Noncommuting projections]
If $P,Q$ are orthogonal projections in a Hilbert space, then $PQ$ is an orthogonal projection if and only if $PQ=QP$.  Hence noncommuting sharp propositions cannot be placed in a single exact Boolean event algebra by using operator multiplication as conjunction.
\end{remark}

\section{Markovian coarse-graining and Choi--Effros infrared ranges}
\label{sec:ce}

We now introduce the dynamical mechanism.  Throughout this section $\A$ is a von Neumann algebra, $\Gamma:\A\to\A$ is a normal UCP map, and $\omega$ is a faithful normal state satisfying $\omega\circ\Gamma=\omega$.

\begin{definition}[Coarse-graining channel]
A Markovian coarse-graining channel is a normal UCP map $\Gamma:\A\to\A$ together with a faithful normal invariant state $\omega$.  Its Cesaro averages are
\[
E_N(A)=\frac1N\sum_{k=0}^{N-1}\Gamma^k(A).
\]
We say that $\Gamma$ has a normal infrared projection if $E_N$ converges point-ultraweakly to a normal UCP idempotent $E_\infty$.
\end{definition}

The infrared operator system is
\[
\A_{\rm IR}:=E_\infty(\A).
\]
By Choi--Effros theory, the product
\[
X\circ Y:=E_\infty(XY),\qquad X,Y\in\A_{\rm IR},
\]
turns $\A_{\rm IR}$ into a $C^*$-algebra \cite{ChoiEffros1976,ChoiEffros1977}.  This product need not be the ambient product of $\A$, and the Choi--Effros algebra need not be commutative.

Let $(\pi_\omega,\Hh_\omega,\Omega_\omega)$ be the GNS representation of $\omega$, and write
\[
\|A\|_\omega=\omega(A^*A)^{1/2}.
\]

\begin{definition}[Infrared asymptotic abelianness]
\label{def:aa}
Assume that $E_N\to E_\infty$ in the GNS seminorm on a norm-dense unital $*$-subspace $\A_{\rm loc}\subset\A$.  The infrared range is asymptotically abelian if
\[
\lim_{N\to\infty}\bigl\|[E_N(A),E_N(B)]\bigr\|_\omega=0,
\qquad A,B\in\A_{\rm loc}.
\]
\end{definition}

This is a mean-infrared formulation.  A pointwise condition such as
$\|[\Gamma^n(A),\Gamma^n(B)]\|_\omega\to0$ implies it under the usual uniform boundedness and mixing estimates used below, but the mean form is the minimal condition needed for the Choi--Effros range.

\begin{theorem}[Boolean infrared range]
\label{thm:boolean-ir}
Let $\Gamma$ be a normal UCP channel preserving a faithful normal state $\omega$.  Suppose the normal infrared projection $E_\infty$ exists and the infrared range is asymptotically abelian in the sense of Definition~\ref{def:aa}.  Then
\[
\pi_\omega(\A_{\rm IR})''
\]
is a commutative von Neumann algebra.  Consequently
\[
\Proj\bigl(\pi_\omega(\A_{\rm IR})''\bigr)
\]
is a complete Boolean algebra.
\end{theorem}

\begin{proof}
For $A,B\in\A_{\rm loc}$, GNS convergence of $E_N$ gives
\[
\pi_\omega(E_N(A))\Omega_\omega\to
\pi_\omega(E_\infty(A))\Omega_\omega,
\qquad
\pi_\omega(E_N(B))\Omega_\omega\to
\pi_\omega(E_\infty(B))\Omega_\omega .
\]
Since $\|E_N(A)\|\le\|A\|$ and $\|E_N(B)\|\le\|B\|$, multiplication is continuous in the GNS seminorm under bounded left and right factors.  Hence
\[
\|[E_\infty(A),E_\infty(B)]\|_\omega
\le
\limsup_{N\to\infty}\|[E_N(A),E_N(B)]\|_\omega=0 .
\]
Thus $\pi_\omega(E_\infty(A))$ and $\pi_\omega(E_\infty(B))$ commute on the cyclic vector in the left ideal seminorm.  Faithfulness of $\omega$ implies that vanishing GNS seminorm is equality in the represented algebra.  Density of $\A_{\rm loc}$ and normal closure give
\[
[\pi_\omega(X),\pi_\omega(Y)]=0,
\qquad X,Y\in \A_{\rm IR}.
\]
Therefore $\pi_\omega(\A_{\rm IR})''$ is commutative.  The final claim follows from Proposition~\ref{prop:boolean-center}.
\end{proof}

\begin{corollary}[Exact infrared events]
If the represented infrared influence algebra is
\[
\I_{\rm IR}:=\pi_\omega(\A_{\rm IR})'',
\]
under the hypotheses of Theorem~\ref{thm:boolean-ir}, then
\[
\Ev_{\rm IR}:=\Proj(Z(\I_{\rm IR}))=\Proj(\I_{\rm IR})
\]
is a complete Boolean algebra.
\end{corollary}

\begin{remark}
The theorem does not say that every fixed point or every Choi--Effros range is classical.  Noncommutative fixed-point algebras and noiseless subsystems are allowed.  They simply do not define exact Boolean records unless their center is selected.
\end{remark}

\section{A block-primitive route to asymptotic abelianness}
\label{sec:block}

The previous theorem isolates the necessary algebraic condition.  We next give a sufficient dynamical criterion.  The formulation is finite-sector for clarity; countable-sector extensions require summability of the constants below.

\begin{assumption}[Block-primitive coarse-graining]
\label{ass:block}
Let $\{P_a\}_{a\in F}$ be a finite family of mutually orthogonal projections in $\A$ with $\sum_aP_a=1$ and $\Gamma(P_a)=P_a$.  For each $a$, let $\omega_a$ be a normal state on $P_a\A P_a$.  There exist constants $\gamma>0$ and, for each local observable $X$, a constant $C_X$ such that
\[
\Gamma^n(X)=
\sum_{a\in F}\omega_a(P_aXP_a)P_a+R_n(X),
\qquad
\|R_n(X)\|_\omega\le C_Xe^{-\gamma n}.
\]
\end{assumption}

The projections $P_a$ label macroscopic sectors.  The estimate contains two ingredients: off-block coherences decay, and within each block the channel is primitive, so local observables converge to scalar multiples of the sector identity.

\begin{theorem}[Block-primitive asymptotic abelianness]
\label{thm:block-primitive}
Under Assumption~\ref{ass:block}, the channel is asymptotically abelian on local observables:
\[
\|[\Gamma^n(A),\Gamma^n(B)]\|_\omega
\le
C_{A,B}e^{-\gamma n}
\]
for all local $A,B$ and a constant $C_{A,B}$.  Moreover the represented infrared von Neumann algebra is generated by the commuting sector projections $\{P_a\}_{a\in F}$.
\end{theorem}

\begin{proof}
Set
\[
X_A^{\rm cl}:=\sum_a\omega_a(P_aAP_a)P_a,
\qquad
X_B^{\rm cl}:=\sum_a\omega_a(P_aBP_a)P_a .
\]
These elements commute because they are scalar combinations of mutually orthogonal projections.  Thus
\[
[\Gamma^n(A),\Gamma^n(B)]
=
[X_A^{\rm cl}+R_n(A),X_B^{\rm cl}+R_n(B)] .
\]
The term $[X_A^{\rm cl},X_B^{\rm cl}]$ vanishes.  The remaining terms contain at least one $R_n$ factor.  Since the classical parts are uniformly bounded by $\|A\|$ and $\|B\|$, the GNS seminorm estimate follows from the triangle inequality and Assumption~\ref{ass:block}.  Taking $n\to\infty$ shows that all local observables have the same infrared representatives as elements of $\vN\{P_a:a\in F\}$ in the GNS representation.  This algebra is commutative.
\end{proof}

\begin{corollary}
Under Assumption~\ref{ass:block}, the hypotheses of Theorem~\ref{thm:boolean-ir} hold whenever the Cesaro infrared projection exists on the local algebra and extends normally.  The infrared event algebra is
\[
\Proj\left(\vN\{P_a:a\in F\}\right)
=
\left\{\sum_{a\in S}P_a:S\subseteq F\right\}.
\]
\end{corollary}

\begin{remark}[Physical reading]
Locality and finite propagation restrict which degrees of freedom can contribute to a coarse macroscopic response; Lieb--Robinson bounds are the standard mathematical form of this intuition in quantum spin systems \cite{LiebRobinson1972}.  ETH and scrambling motivate primitive mixing within high-dimensional nonintegrable sectors \cite{Deutsch1991,Srednicki1994}.  The theorem uses these ideas only as motivation for Assumption~\ref{ass:block}; the algebraic conclusion depends on the displayed estimate.
\end{remark}

\begin{remark}[Nonabelian residual sectors]
If a channel contains a protected nonabelian noiseless subsystem, Assumption~\ref{ass:block} fails.  The corresponding infrared range may remain noncommutative.  In the present framework this is not a contradiction: such a sector is residual quantum information, not an exact Boolean record algebra.
\end{remark}

\section{Split-record approximants and RG stability}
\label{sec:split-rg}

AQFT split inclusions provide a useful finite-stage comparison, but they are not taken here as fundamental sources of classicality.  Under standard split hypotheses, local regions $O\Subset\widetilde O$ may admit
\[
\A(O)\subset \F(O,\widetilde O)\subset \A(\widetilde O),
\]
where $\F$ is a Type-I factor \cite{DoplicherLongo}.  A Type-I factor still has trivial center.  Nontrivial records enter only after a record decomposition
\[
\Hh_{\rm lab}=\bigoplus_{a\in F}\Hh_a,
\qquad
P_a:\Hh_{\rm lab}\to\Hh_a,
\]
and the conditional expectation
\[
E_R(X)=\sum_{a\in F}P_aXP_a .
\]
The range is
\[
\R=\bigoplus_{a\in F}B(\Hh_a),
\qquad
Z(\R)=\left\{\sum_{a\in F}\lambda_aP_a\right\}.
\]
Thus split-record laboratories are finite-stage models of the block structure in Section~\ref{sec:block}; the asymptotic theorem explains when this structure is generated by coarse-graining.

The same language clarifies when exact events are stable under renormalization.  Let $\Gamma_\ell:\I_{\ell+1}\to\I_\ell$ be a normal UCP map between influence algebras.  Its multiplicative domain is
\[
\MD(\Gamma_\ell)=
\{a:\Gamma_\ell(a^*a)=\Gamma_\ell(a)^*\Gamma_\ell(a),\
\Gamma_\ell(aa^*)=\Gamma_\ell(a)\Gamma_\ell(a)^*\}.
\]

\begin{lemma}[Projection preservation]
\label{lem:md-proj}
If $P=P^*=P^2$, then $P\in\MD(\Gamma_\ell)$ if and only if $\Gamma_\ell(P)$ is a projection.
\end{lemma}

\begin{proof}
For a projection $P$, the two multiplicative-domain identities become
\[
\Gamma_\ell(P)=\Gamma_\ell(P)^*\Gamma_\ell(P),
\qquad
\Gamma_\ell(P)=\Gamma_\ell(P)\Gamma_\ell(P)^* .
\]
Since positive maps preserve adjoints on self-adjoint elements, this is equivalent to $\Gamma_\ell(P)^2=\Gamma_\ell(P)$.
\end{proof}

\begin{proposition}[Boolean RG homomorphism]
\label{prop:boolean-rg}
Assume
\[
Z(\I_{\ell+1})\subseteq \MD(\Gamma_\ell),
\qquad
\Gamma_\ell(Z(\I_{\ell+1}))\subseteq Z(\I_\ell).
\]
Then
\[
\Gamma_\ell:\Proj(Z(\I_{\ell+1}))\longrightarrow \Proj(Z(\I_\ell))
\]
is a Boolean algebra homomorphism.
\end{proposition}

\begin{proof}
Lemma~\ref{lem:md-proj} gives projection preservation.  For central projections $P,Q\in Z(\I_{\ell+1})$, multiplicativity on the multiplicative domain gives
\[
\Gamma_\ell(PQ)=\Gamma_\ell(P)\Gamma_\ell(Q).
\]
Unitality gives $\Gamma_\ell(1-P)=1-\Gamma_\ell(P)$, and linearity then gives preservation of finite joins.  The second hypothesis places the images in the target center.
\end{proof}

\section{Scope and limitations}
\label{sec:limits}

The construction deliberately separates definitions, theorems, and physical interpretation.

First, exact events are defined as center projections of an influence algebra.  This is a conservative algebraic definition, not a claim that every projection in a microscopic local algebra is a fact.  It is compatible with the Type-III character of AQFT local algebras because nontrivial exact records arise only after taking the appropriate response-stabilized infrared algebra.

Second, Choi--Effros theory alone does not imply classicality.  A completely positive infrared range can remain noncommutative.  The additional condition is asymptotic abelianness in the relevant representation.  The block-primitive theorem gives one sufficient route to this condition, but it is not the only possible route.

Third, the framework has clear failure modes.  Integrable dynamics, many-body localization, exact nonabelian noiseless subsystems, or topological memory can prevent the infrared algebra from becoming abelian.  In such cases the theory predicts residual quantum sectors rather than exact Boolean records.

The mathematical conclusion is therefore modest but useful: under explicit state-preserving CP coarse-graining assumptions, macroscopic Boolean facts can be represented as the projections of a commutative GNS infrared von Neumann algebra generated dynamically from noncommutative quantum operations.

\section*{Acknowledgements}

The author thanks colleagues and readers of the RQCP-QG preprint for discussions on operator-algebraic formulations of coarse-graining and exact records.

\section*{Declarations}

\textbf{Funding:} The author declares that no funds, grants, or other support were received during the preparation of this manuscript. \\
\textbf{Competing Interests:} The author has no relevant financial or non-financial interests to disclose. \\
\textbf{Data Availability:} Data sharing is not applicable to this article as no datasets were generated or analysed during the current study.

\end{document}